\documentclass[conference]{IEEEtran}
\usepackage{bigints}
\usepackage{caption}
\newlength\figureheight 
\newlength\figurewidth 
\usepackage{cite}
\usepackage{url}
\usepackage{blindtext, graphicx}
\usepackage{amsmath,amssymb,mathtools,bm,etoolbox}
\newcommand\Mark[1]{\textsuperscript#1}

\newcommand{\e}[1]{{\mathbb E}\left[ #1 \right]}
\usepackage{stackengine}
\def\delequal{\mathrel{\ensurestackMath{\stackon[1pt]{=}{\scriptstyle\Delta}}}}
\usepackage{suffix}

\DeclarePairedDelimiterX\MeijerM[3]{\lparen}{\rparen}%
{\begin{smallmatrix}#1 \\ #2\end{smallmatrix}\delimsize\vert\,#3}

\newcommand\MeijerG[8][]{%
  G^{\,#2,#3}_{#4,#5}\MeijerM[#1]{#6}{#7}{#8}}

\WithSuffix\newcommand\MeijerG*[7]{%
  G^{\,#1,#2}_{#3,#4}\MeijerM*{#5}{#6}{#7}}
\begin{document}

\title{Hybrid Rayleigh and Double-Weibull over Impaired RF/FSO System with Outdated CSI}

\author{Elyes~Balti\Mark{1},~Mohsen~Guizani\Mark{1}~and Bechir~Hamdaoui\Mark{2}\\
        \Mark{1}University of Idaho, USA,  \Mark{2}Oregon State University, USA}

\maketitle

\begin{abstract}
In this work, we present a global framework of a dual-hop RF/FSO system with multiple relays operating at the mode of amplify-and-forward (AF) with fixed gain. Partial relay selection (PRS) protocol with outdated channel state information (CSI) is assumed since the channels of the first hop are time-varying. The optical irradiance of the second hop are subject to the Double-Weibull model while the RF channels of the first hop experience the Rayleigh fading. The signal reception is achieved either by heterodyne or intensity modulation and direct detection (IM/DD). In addition, we introduce an aggregate model of hardware impairments to the source (S) and the relays since they are not perfect nodes. In order to quantify the impairment impact on the system, we derive closed-form, approximate, upper bound and high signal-to-noise ratio (SNR) asymptotic of the outage probability (OP) and the ergodic capacity (EC). Finally, analytical and numerical results are in agreement using Monte Carlo simulation.
\end{abstract}


\IEEEpeerreviewmaketitle

\section{Introduction}
Free space optical (FSO) communications have recently gained enormous interest for many applications such as back-haul for wireless cellular network, disaster recovery and redundent link \cite{1} since it provides free access to the spectrum and high bandwidth. These advantages make the FSO technique not only the corner stone of wireless 5G but also as a complementary to the RF communication. In fact, RF communication reaches its bottleneck since it suffers from spectrum drought, expensive spectrum access, susceptibility to the interferences and network attacks \cite{2}. Unlike the RF mode, FSO communication is characterized by a high security level and a minimum of bit/symbol error rate. Although previous research attempts have confirmed the reliability of FSO technology, it has some limitations mainly caused by the atmospheric turbulences also called optical fading or scintillation. These fluctuations are originated by the variations of the refractive index of the propagation medium due to the heterogeneity in temperature and the fluctuations of the atmospheric pressures. Given that the optical signal is very sensitive to these turbulences, many previous attempts were interested in modelling the optical irradiance in order to quantify the fading impact on the system performance with accuracy. The first model proposed in this context is Log-Normal distribution which is dedicated to describe the weak turbulences. As the atmospheric turbulences become more severe, this model deviates from the experimental results. To overcome this discrepancy, recent work have proposed a more sophisticated model to describe moderate and strong fading called Double-Gamma. In fact, this model is widely used in recent work of mixed RF/FSO systems \cite{3}, \cite{4} since it provides more accurate performance metrics (outage probability, average bit/symbol error rate, ergodic capacity) than Log-Normal model. Despite these advantages, Double-Gamma suffers from either overestimation and underestimation in the tail region of the probability density function (PDF). To overcome this problem, Nestor \textit{et al}. \cite{5} developed an advanced irradiance model called Double-Weibull which is more precise and abide to the experimental data not only around centralized data but also in the tail data region. To achieve better performance, it is recommended to introduce this model into the mixed RF/FSO cooperative relaying systems which have recently attracted considerable attention since it provides not only better QoS and coverage probability but also enhances the system capacity. Recent research attempts have addressed many relaying schemes that can be implemented into the relays. These protocols are mainly Decode-and-Forward (DF) \cite{7}, \cite{8}, Amplify-and-Forward (AF) \cite{9}, \cite{10}, Quantify-and-Encode (QE) \cite{11}, \cite{12}. In practice, due to its low quality, the hardwares (source, relays ...) are susceptible to the impairments e.g., non-linear high power amplifier (HPA) \cite{17}, \cite{18}, phase noise \cite{19} and I/Q imbalance \cite{20}. Schenk \textit{et al}. \cite{21} have proven that the I/Q imbalance rotates the phase constellation and attenuates the amplitude while Dardari \textit{et al}. \cite{17} have concluded that the HPA non-linearities creates non-linear distortion during the signal amplification. Furthermore, many research attempts \cite{17}, \cite{21} and \cite{24} turned out that the system capacity is limited by a ceiling created by the joint effect of non-linear HPA and I/Q imbalance. Given that these impairments can be neglected for low rate systems, the imperfection impacts become significant for high rate systems and especially as the average SNR largely increases. Our contribution is to propose the Double-Weibull as a model for the optical irradiance of the second hop of the mixed FR/FSO system with multiple relays. In addition, we assume partial relay selection with outdated CSI \cite{13} \cite{4}, \cite{16} to select one relay among the sets. Although, the PRS is less performant than the opportunistic relay selection, it is more efficient in terms of power consumption and complexity. To generalize this work, we introduce an aggregate model of impairments to the source and the relays as the work proposed by \cite{f} where they introduced a general model of impairment to a mixed RF/FSO system assuming the Double-Gamma as a model for the irradiance. To the best of our knowledge, we are the first research team who propose a mixed RF/FSO systems with multiple relays where the source and the relays are affected by a general model of impairment. The RF channels are modelled by correlated Rayleigh while the FSO channels are subject to the Double-Weibull fading. Moreover, the signal is received either by heterodyne or IM/DD detection methods. The rest of this paper is organized as follows: section II presents the system and the channels models while the outage probability and the ergodic capacity analysis are detailed in section III. Analytical and numerical results following their discussion are presented in section IV. Finally, concluding remarks and future directions are reported in section V.
\section{System and Channels Models}
\subsection{System Model}
The system consists of $S$, $D$ and $N$ parallel relays wirelessly linked to $S$ and $D$. In order to pick a relay of rank $m$, PRS with outdated CSI based on the partial knowledge of the CSIs channels of the first hop is assumed. This protocol states that for each transmission, $S$ receives the CSIs ($\gamma_{1(l)}$ for $l$ = 1,... $N$) of the RF channels from the relays via local feedback. Once the CSIs are received, $S$ sorts the values of the CSIs in an increasing order of amplitude as: $\gamma_{1(1)}\leq\gamma_{1(2)}\leq \ldots \leq\gamma_{1(N)}$. Based on this sorting, $S$ selects the relay with the highest RF SNR which is clearly the relay of last rank $N$. Given that the relays operate at the half-dupplex mode, the best relay of rank $N$ may not be always available to forward the signal. In this case, $S$ will select the next best relay and so on so forth. In addition, the relay with the last rank is not always the best one even after the selection. In fact, the channels are time-varying and the feedback propagation from the relays to $S$ are very slow. In this case, the CSIs are suscpetible to significant variations and so their values before and after the selection are not the same. It turned out that the estimation of the channels is not perfect and hence, the relay selection is achieved based on the outdated CSIs. To model this imperfect channel estimation, we associate a time correlation coefficient $\rho$ between the outdated and the updated CSIs. Thereby, the best relay is not necessarily the one of the last rank since the selection is based on the outdated CSI.\\
Assume that $S$ selects the relay of rank $m$, the received signal at the relay is given by:
\begin{equation}
    y_{1(m)} = h_m (s + \eta_1) + \nu_1
\end{equation}
where $h_m$ is the RF channel fading, $s \in \mathbb{C}$ is the information signal, $\nu_1$  $\backsim$ $\mathcal{CN}$ (0, $\sigma^2_0$) is the AWGN of the RF channel, $\eta_1$ $\backsim$ $\mathcal{CN}$ (0, $\kappa^2_{1}P_1$) is the distortion noise at $S$, $\kappa_{1}$ is the impairment level at $S$ and $P_1$ is the average transmitted power from $S$.\\
Once the signal is completely received by the relay $R_m$, it is amplified by a fixed gain $G$ that depends on the average electrical SNR of the RF channels. This gain can be expressed as follows [20, eq.~(11)]:
\begin{equation}
G^2 \delequal \frac{P_2}{P_1\e{|h_m|^2}(1+\kappa^2_1)+\sigma^2_0}
\end{equation}
where $P_2$ is the average transmitted power from the relay to $D$ and $\e{.}$ is the expectation operator.\\
The amplified signal at the output of the relay is given by:
\begin{equation}
y_{opt(m)} = G (1 + \eta_e) y_{1(m)}
\end{equation}
where $\eta_e$ is the electrical-to-optical conversion coefficient.\\
Finally the received signal at the destination can be expressed as follows:
\begin{equation}
y_{2(m)} = (\eta_o I_m)^{\frac{r}{2}} [G (1 + \eta_e)(h_m (s + \eta_1) + \nu_1) + \eta_{2}] + \nu_2 
\end{equation}
where $\eta_o$ is the optical-to-electrical conversion coefficient, $I_m$ is the optical irradiance between R$_m$ and $D$, $\eta_{2}$ $\backsim$ $\mathcal{CN}$ (0, $\kappa^2_{2}P_2$) is the distortion noise at the relay R$_m$, $\kappa_{2}$ is the impairment level at R$_m$, $\nu_2$  $\backsim$ $\mathcal{CN}$ (0, $\sigma^2_0$) is the AWGN of the optical channel and $r$ = 1, 2 stands for heterodyne and IM/DD detections respectively.
\subsection{End-to-End Signal-to-Noise plus Distortion Ratio (SNDR)}
The SNDR depends on both the electrical $\gamma_{1(m)}$ and  optical $\gamma_{2(m)}$ SNRs of the two hops which can be defined by:
\begin{equation}
\gamma_{1(m)} = \frac{|h_m|^2P_1}{\sigma^2_0} = |h_m|^2\overline{\gamma}_1
\end{equation}
where $\overline{\gamma}_1 = \frac{P_1}{\sigma^2_0}$ is the average SNR of the first hop.
\begin{equation}
\gamma_{2(m)} = \frac{|I_m|^r\eta_{o}^rP_2}{\sigma^2_0} = |I_m|^r\overline{\gamma}_r
\end{equation}
where $\overline{\gamma}_r = \frac{\eta_{o}^r P_2}{\sigma^2_0}$ is the average electrical SNR of the second hop.
Finally, the SNDR can be expressed as follows:
\begin{equation}
\gamma_{\text{ni}} = \frac{|h_m|^2 |I_m|^r}{\delta |h_m|^2|I_m|^r + |I_m|^r(1 + \kappa^2_{2})\frac{\sigma^2_0}{P_1} + \frac{\sigma^2_0}{P_1 G^2}} 
\end{equation}
After some algebraic manipulations, the SNDR can be expressed as follows:
\begin{equation}
\gamma_{\text{ni}} = \frac{\gamma_{1(m)}\gamma_{2(m)}}{\delta \gamma_{1(m)}\gamma_{2(m)} + (1 + \kappa^2_{2})\gamma_{2(m)} + C}
\end{equation}
where $\delta \delequal \kappa^2_{1} + \kappa^2_{2} + \kappa^2_{1}\kappa^2_{2}$ and $C = \e{\gamma_{1(m)}}(1+\kappa_1^2)+1$.\\
Note that for ideal case, the end-to-end SNR is given by:
\begin{equation}
\gamma_{\text{id}} = \frac{\gamma_{1(m)}\gamma_{2(m)}}{\gamma_{2(m)} + \e{\gamma_{1(m)}} +1 }
\end{equation}
\subsection{Channels Model}
\subsubsection{Statistics of the electrical channels}: We model the relation between the outdated and updated CSIs as follows:
\begin{equation}
h_{1(m)} = \sqrt{\rho} \tilde{h}_{1(m)} + \sqrt{1-\rho} w_{1(m)} 
\end{equation}
where $h_{1(m)}$ and $\tilde{h}_{1(m)}$ are the updated and outdated CSIs respectively and $w_{1(m)}$ follows the circularly symmetric complex gaussian distribution with the same variance of the channel gain $\tilde{h}_{1(m)}$.\\
The coefficient $\rho$ is given by the Jakes' autocorrelation model \cite{39} as follows:
\begin{equation}
\rho = J_0(2\pi f_{d} T_d) 
\end{equation}
where, $J_0(.)$ is the zeroth order Bessel function of the first kind eq. (8.411) in \cite{37}; $T_d$ is the time delay between the current CSI and the delayed version and $f_d$ is the maximum Doppler frequency of the channels.\\
Since the RF channels experience the Rayleigh fading, the instantaneous electrical SNR $\gamma_{1(m)}$ follows the correlated exponential distribution. The PDF can be expressed as follows:
\begin{equation}
\begin{split}
f_{\gamma_{1(m)}}(x) = \sum_{n=0}^{m-1} {m-1 \choose n} 
\frac{(-1)^n}{[(N-m+n)(1-\rho)+1]\overline{\gamma}_1}\\\times~m{N \choose m} \exp\left(-\frac{ (N-m+n+1)x }{[(N-m+n)(1-\rho)+1]\overline{\gamma}_1}\right)
\end{split}
\end{equation}
After some mathematical manipulations, the CDF of $\gamma_{1(m)}$ can be expressed as follows:
\begin{equation}
\begin{split}
F_{\gamma_{1(m)}}(x) = 1 - m{N \choose m}\sum_{n=0}^{m-1} \frac{(-1)^n}{N-m+n+1} \\ 
\times \exp\left(-\frac{ (N-m+n+1)x }{[(N-m+n)(1-\rho)+1]\overline{\gamma}_1}\right)
\end{split}
\end{equation}
The constant $C$ mentioned earlier depends on the expression of $\e{\gamma_{1(m)}}$, which can be obtained as:
\begin{equation}
\begin{split}
\e{\gamma_{1(m)}} = m{N \choose m}\sum_{n=0}^{m-1} {m-1 \choose n} (-1)^n\\ \times~\frac{[(N-m+n)(1-\rho)+1]\overline{\gamma}_1}{(N-m+n+1)^2}  
\end{split}
\end{equation}
\subsubsection{Statistics of the optical channels}:
The PDF of the random variable $X$ that follows the Weibull distribution can be written as follows:
\begin{equation}
f_X(x) = \frac{\beta_1 x^{\beta_1-1}}{\Omega_1}~\exp\left(-\frac{x^{\beta_1}}{\Omega_1} \right)
\end{equation}
where $\Omega_1 > 0$ is the average fading power of the optical fading and $\beta_1 > 0$ describes the strength of the irradiance fluctuations.\\
According to the scintillation theory, it is possible to model the irradiance as the product of two independent random variables $X, Y$  following the Weibull distribution.
Since the irradiance is modelled by the Double-Weibull, the PDF of $I = XY$ can be obtained by \cite{5}, eq.~(5):
\begin{equation}
\begin{split}
f_{I}(I) = \frac{\beta_2k\sqrt{kl}}{(2\pi)^{\frac{k+l}{2}-1}I}
\MeijerG[\Bigg]{0}{k+l}{k+l}{0}{\Lambda_0}{-}{\left(\frac{\Omega_2k}{I^{\beta_2}}\right)^k(\Omega_1l)^l}
\end{split}
\end{equation}
where $\Lambda_0$ is given by:
\begin{equation*}
\Lambda_0 = [\Delta(l;0), \Delta(k;0)]
\end{equation*}
$G^{m,n}_{p,q}(.)$ is the Meijer's G-function, $\Delta(j;x) \delequal \frac{x}{j}, \ldots , \frac{x+j-1}{j}$ and $l, k$ are positive integers satisfying:
\begin{equation}
\frac{l}{k} = \frac{\beta_2}{\beta_1}
\end{equation}
where $\beta_1, \beta_2 > 0$ are the parameters describing the strength of the optical irradiance from large and small scale turbulent eddies. In addition, $\Omega_1, \Omega_2 > 0$ are the average power of the channels.\\
The CDF of the optical irradiance can be expressed as follows:
\begin{equation}
\begin{split}
 F_{I}(I) = \frac{\sqrt{kl}}{(2\pi)^{\frac{k+l}{2}-1}}
\MeijerG[\Bigg]{k+l}{1}{1}{k+l+1}{1}{\Lambda_1, 0}{\frac{I^{\beta_1l}}{(\Omega_1l)^l(\Omega_2k)^k}}   
\end{split}
\end{equation}
where $\Lambda_1$ is given by:
\begin{equation*}
    \Lambda_1 = [\Delta(l;1), \Delta(k;1)]
\end{equation*}
The normalized variances $\sigma_i^2$ and the average fading powers $\Omega_i$ of the large and small scale atmospheric turbulence are given by:
\begin{equation*}
\sigma_i^2 = \frac{\Gamma(1 + \frac{2}{\lambda_i})}{\Gamma(1 + \frac{2}{\lambda_i})^2} - 1,~  
\Omega_i = \left(\frac{1}{\Gamma(1+\frac{1}{\beta_i})}\right)^{\beta_i}
\end{equation*}
where $i = 1, 2$ and $\sigma^2_X = \sigma^2_1, \sigma^2_Y = \sigma^2_2$. $\lambda_i$ can be determined by $\lambda_i = \sigma_i^{-1.0852}$.\\
Now, we substitute the analytical expression of the optical channel $I_m$ by $\left(\frac{\gamma_{2(m)}}{\overline{\gamma}_r}\right)^{\frac{1}{r}}$ in eq.~(18) 
and after some mathematical manipulations, the PDF and the CDF of the instantaneous SNR $\gamma_{2(m)}$ can be respectively written as follows:
\begin{equation}
\begin{split}
f_{\gamma_{2(m)}}(\gamma_{2(m)}) = \frac{\beta_2k\sqrt{kl}}{r\gamma_{2(m)}(2\pi)^{\frac{k+l}{2}-1}}~~~~~~~\\ \times~
\MeijerG[\Bigg]{0}{k+l}{k+l}{0}{\Lambda_0}{-}{(\Omega_1l)^l(\Omega_2k)^k\left(\frac{\overline{\gamma}_r}{\gamma_{2(m)}}\right)^{\frac{\beta_2k}{r}}}   
\end{split}
\end{equation}
\begin{equation}
\begin{split}
F_{\gamma_{2(m)}}(\gamma_{2(m)}) = \frac{\sqrt{kl}}{(2\pi)^{\frac{k+l}{2}-1}}~~~~~~~~~~~~~~~\\ \times~
\MeijerG[\Bigg]{k+l}{1}{1}{k+l+1}{1}{\Lambda_1, 0}{\frac{1}{(\Omega_1l)^l(\Omega_2k)^k}\left(\frac{\gamma_{2(m)}}{\overline{\gamma}_r}\right)^{\frac{\beta_1l}{r}}}   
\end{split}
\end{equation}
The $n$-th moment of the random variable $X$ is given by:
\begin{equation}
\e{X^n} = \Omega_1^\frac{n}{\beta_1}\Gamma\left(1 + \frac{n}{\beta_1}\right)
\end{equation}
After some mathematical manipulations, the $n$-th moment of the instantaneous SNR $\gamma_{2(m)}$ can be written as follows:
\begin{equation}
\e{\gamma_{2(m)}^n} = \overline{\gamma}_r^n\Omega_1^{\frac{nr}{\beta_1}}\Omega_2^{\frac{nr}{\beta_2}}\Gamma\left(1 + \frac{nr}{\beta_1} \right)\Gamma\left(1 + \frac{nr}{\beta_2} \right)
\end{equation}
\section{Performance Analysis}
In this section we present the analysis of the system performance in terms of the OP and EC. We will derive the expressions of the OP and the upperbound of the EC in terms of the Meijer's G-function. We will also evaluate the system performance in particular at the high SNR regime and we will show that the EC and the SNDR are saturated by the ceilings created by the hardware impairments.
\subsection{Outage Probability Analysis}
The outage probability is defined as the probability that the end-to-end SNDR falls below an outage threshold $\gamma_{\text{th}}$. It can be written as follows:
\begin{equation}
P_{\text{out}}(\gamma_{\text{th}}) \delequal \text{Pr}\{\gamma_{\text{ni}} < \gamma_{\text{th}}\}
\end{equation}
where \text{Pr(.)} is the probability notation. Then, we substitute the expression of the SNDR in eq.~(21) and after applying some mathematical manipulations, the OP can be written as follows:
\begin{equation}
\begin{split}
P_{\text{out}}(\gamma_{\text{th}}) = \int_{0}^{\infty} F_{\gamma_{1(m)}}\left(\frac{(1+\kappa_2^2)\gamma_{\text{th}}}{1-\delta \gamma_{\text{th}}} + \frac{C\gamma_{\text{th}}}{(1-\delta \gamma_{\text{th}})\gamma_{2(m)}}   \right)~~~~~~~~\\ \times~f_{\gamma_{2(m)}}(\gamma_{2(m)})~d\gamma_{2(m)}~~~~~~~~~~~~~~~~~~~~~~~~~~~
\end{split}
\end{equation}
Note that the CDF $F_{\gamma_{1(m)}}$ is defined only if $1 - \delta\gamma_{\text{th}} > 0$, otherwise it is equal to a unity. Using the identity given by \cite{26}, eq.~(2.24.3.1) and after some mathematical manipulations, the OP can be derived as follows:
\begin{equation}
\begin{split}
P_{\text{out}}(\gamma_{\text{th}}) = 1 - {N \choose m}\frac{mk\sqrt{\beta_2l}r^{\mu - 1}}{(2\pi)^{\frac{\beta_2l+r(k+l)-3}{2}}} \sum_{n=0}^{m-1}\frac{(-1)^n}{N-m+n+1}~~~~~~~~~~~~~~~\\ \times~\exp\left(-\frac{(N-m+n+1)(1+\kappa_2^2)\gamma_{\text{th}}}{[(N-m+n)(1-\rho)+1](1-\delta\gamma_{\text{th}})\overline{\gamma}_1} \right)~~~~~~~~~~~~~\\ \times~
{m-1 \choose n}~\MeijerG[\Bigg]{0}{r(k+l)+\beta_2k}{r(k+l)+\beta_2k}{0}{\Lambda_2}{-}{\zeta}~~~~~~~~~~~~~~~~~~~~~~~~~~~~~~~~
\end{split}
\end{equation}
where $\mu, \Lambda_2$ and $\zeta$ are respectively given by:
\begin{equation*}
\begin{split}
\mu = - \sum_{j=0}^{k+l} \Lambda_0(j) + \frac{k+l}{2} +1~~~~~~~~~~~~~~~~~~~~~~~~\\
\Lambda_2 = [\Delta(r;\Lambda_0), \Delta(\beta_2k;1)]~~~~~~~~~~~~~~~~~~~~~~~~~~~~\\
\zeta = \left((\Omega_1l)^l(\Omega_2k)^kr^{k+l}\right)^r \left(\frac{\beta_2k\overline{\gamma}_1\overline{\gamma}_r}{\tau \xi}\right)^{\beta_2k}~~~~~~~~
\end{split}
\end{equation*}
$\tau$ and $\xi$ are obtained by:
\begin{equation*}
\begin{split}
\tau = \frac{C\gamma_{\text{th}}}{1- \delta \gamma_{\text{th}}},~~ 
\xi = \frac{N-m+n+1}{(N-m+n)(1-\rho)+1}
\end{split}
\end{equation*}
The OP is equal to eq.~(25) for $\gamma_{\text{th}} < \frac{1}{\delta}$, otherwise, it is equal to a unity.
\subsection{Ergodic Capacity Analysis}
The ergodic capacity, expressed in bps/Hz, is defined as the maximum error-free data rate transferred by the system channel. It can be written as follows:
\begin{equation}
\overline{C} = \e{\log_2(1 + c\gamma_{\text{ni}})}
\end{equation}
where $c = 1$ indicates the heterodyne detection and $c = \frac{e}{2\pi}$ for IM/DD. The ergodic capacity can be derived by evaluating the PDF of the SNDR. However, an exact analytical expression of eq.~(26) is not solvable. To evaluate the ergodic capacity, a numerical evaluation is required.\\
It is possible to derive a simpler form of an upper bound which is given by the following theorem.
\newtheorem{theorem}{Theorem}
\begin{theorem}
For Asymmetric (Rayleigh/Double-Weibull) channels, the system capacity $\overline{C}$ with AF relaying protocol and hardware impairments is upper bounded by:
\begin{equation}
    \overline{C} \leq \log_{2}\left(1 + c\frac{\mathcal{J}}{\mathcal{J}\delta + 1}\right)
\end{equation}
\end{theorem}
where $\mathcal{J}$ is given by:
\begin{equation}
\mathcal{J} = \frac{\beta_2k\sqrt{kl}r^{\mu-1}\e{\gamma_{1(m)}}}{(2\pi)^{\beta_2k+r\frac{k+l}{2}-2}(1+\kappa_2^2)} \MeijerG[\Bigg]{\beta_2k}{r(k+l)+\beta_2k}{r(k+l)+\beta_2k}{\beta_2k}{\Lambda_2}{\Delta(\beta_2k;1)}{\varrho}
\end{equation}
where $\varrho$ is given by:
\begin{equation*}
\begin{split}
\varrho = ((\Omega_1l)^l(\Omega_2k)^kr^{k+l})^r\left(\frac{(1+\kappa_2^2)\overline{\gamma}_r}{C}\right)^{\beta_2k}
\end{split}
\end{equation*}
Although deriving a closed-form of the ergodic capacity is very complex, we can find an approximate simpler form by applying the approximation given by \cite{25}, eq.~(35):
\begin{equation}
\e{\log_2\left(1 + \frac{\psi}{\varphi}\right)} \approx \log_2\left(1+\frac{\e{\psi}}{\e{\varphi}} \right)
\end{equation}
For high SNR regime, the behavior of the SNDR is expressed as:
\begin{equation}
\lim_{\overline{\gamma}_1,\overline{\gamma}_r\to\infty}\gamma_{\text{ni}} = \frac{1}{\delta}
\end{equation}
We observe that the SNDR converges to a ceiling $\gamma^* = \frac{1}{\delta}$.
\newtheorem{thm}{Theorem}
\newtheorem{cor}[thm]{Corollary}
\begin{cor}
For larger values of $\overline{\gamma}_1$ and $\overline{\gamma}_r$ and mutually independent RF and optical fadings, the average channel capacity converges to a ceiling defined by $C^* = \log_2(1+c\gamma^*)$. 
\end{cor}
\begin{proof}
Applying the dominated convergence theorem and given that the SNDR is limited by $\gamma^*$, the limit can be moved inside the logarithm function as shown below:
\begin{equation}
\begin{split}
\lim_{\overline{\gamma}_1,\overline{\gamma}_r\to\infty}\log_2(1 + c\gamma_{\text{ni}}) = \log_2(1 + c\lim_{\overline{\gamma}_1,\overline{\gamma}_r\to\infty}\gamma_{\text{ni}})\\
= \log_2(1+c\gamma^*)~~~~~~~~~~~
\end{split}
\end{equation}
\end{proof}
\section{Numerical Results}
This section provides numerical results obtained by using the mathematical formulations of the previous section.\\
The electrical channel is subject to the correlated Rayleigh fading which can be generated using the algorithm in \cite{28}. The atmospheric turbulence is modeled by the Double-Weibull fading, which can be generated by using the formula, $I = XY$, where $X$ and $Y$ are mutually independent Weibull random variables.
\vspace*{-0.5cm}
\begin{center}
\includegraphics[width=9cm,height=6.5cm]{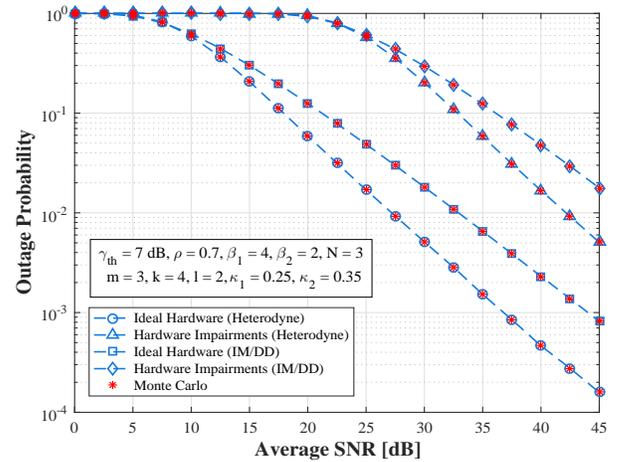}
\captionof{figure}{Outage probability versus the average SNR for ideal and non-ideal hardware under IM/DD and heterodyne detection}
\label{fig1}
\end{center}
Fig.~1 shows the dependence of the OP with respect to the average SNR. As proved by previous work, we note that the heterodyne detection outperforms the IM/DD method for our system. Moreover, the impact of the hardware impairments are clearly observed compared to the ideal hardware case. For low SNR, the impairments have small impact on the performance and so it can be neglected as we mentioned earlier. As the average SNR increases, the impact of the impairments becomes more severe enough to be of high importance and must not be neglected.
\begin{center}
\includegraphics[width=9cm,height=6.5cm]{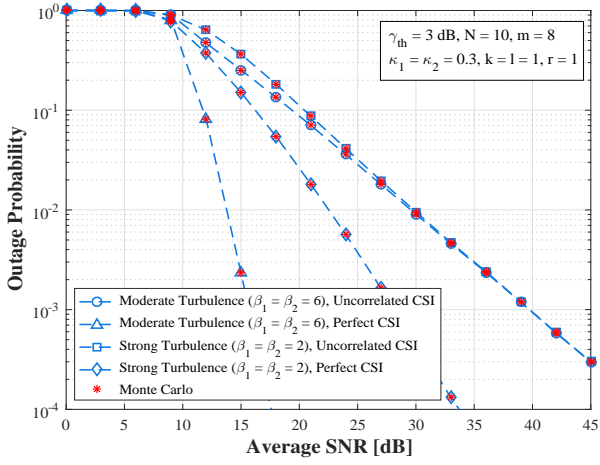}
\captionof{figure}{Outage probability versus the average SNR for different values of the correlation coefficient $\rho$ under moderate and strong turbulences}
\label{fig1}
\end{center}
\begin{center}
\includegraphics[width=9cm,height=6.5cm]{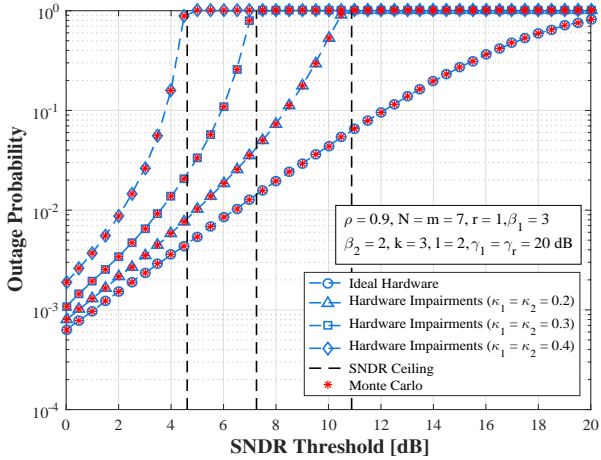}
\captionof{figure}{Outage  probability versus the SNDR threshold for ideal and non-ideal hardware}
\label{fig1}
\end{center}
The dependence of the OP for AF relaying protocol on the average SNR is given by Fig.~2. As expected, the outage performance is better under the moderate turbulence condition and suddenly deteriorate as the turbulence becomes strong and severe. This result is clearly observed, especially for the case of full correlation of CSIs ($\rho$ = 1). It turned out that the system substantially depends on the state of the optical channels. As the correlation $\rho$ between the CSI used for relay selection and the CSI used for transmission increases, i.e., the two CSIs become more and more correlated, the selection of the best relay is certainly achieved ($m$ = $N$). In this case, the system works under the perfect condition especially under moderate turbulence condition. As the time correlation decreases, the selection of the best relay is no longer achieved and so the system certainly operates with a worse relay. In addition, we note that the correlation has a severe impact on the performance. In fact, for the case of completely outdated CSI ($\rho$ = 0), we observe a substantial degradation of the performance for moderate and strong turbulence conditions and the curves most likely look the same. In other words, considering either moderate or strong turbulence conditions has no remakable impact on the performance in case of uncorrelated CSIs. This observation proves that the system depends to a large extent on the correlation between the CSIs rather than the state of the optical channels. This is nothing but to say that is important to achieve perfect CSI channels estimation than to focus on the atmospheric turbulence conditions.\\
Fig.~3 presents the variations of the outage probability versus the outage threshold $\gamma_{\text{th}}$ [dB] for different values of the level of impairments ($\kappa_1, \kappa_2$). For lower values of $\gamma_{\text{th}}$, the performance under the hardware impairments slightly deviates from the case of ideal system. However, as the outage threshold increases, the outage performance experiences a rapid convergence to a unity and this convergence becomes more faster as the impairment level grows up. In fact, we observe that for the given values of the impairment level 0.2, 0.3 and 0.4, the system saturates at the following SNDR thresholds 4.6, 7.5 and 10.8 dB respectively, while the ideal system saturates very slowly for an outage threshold greater than 20 dB.
\begin{center}
\includegraphics[width=9cm,height=6.5cm]{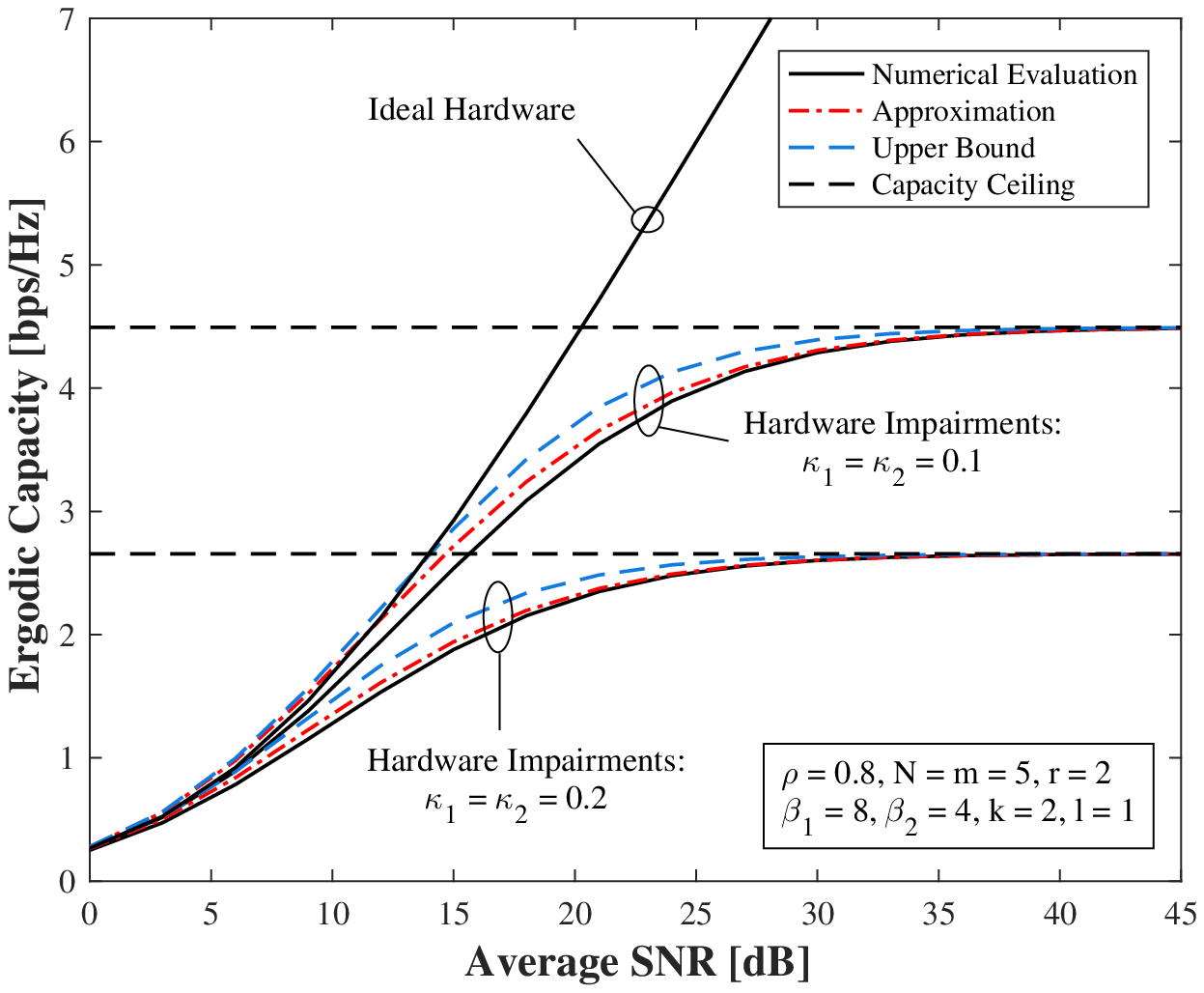}
\captionof{figure}{Outage  probability versus the SNDR threshold for ideal and non-ideal hardware}
\label{fig1}
\end{center}
The variations of the EC versus the average SNR is given by Fig.~4. For the ideal hardware case, as the SNR increases, the EC grows indefinitely. Regarding the non-ideal harware, the impairments have small impact on the system for low SNRs, but it becomes very deleterious at high SNRs. In fact, the EC converges to a capacity ceiling $C^*$, as shown by corollary 1, which is inversely proportional to the level of the impairments, i.e, as the impairment level increases, the ceiling decreases. The approximate form and the upper bound of the EC are also shown in Fig.~4. Although they deviate from the exact EC at low SNRs, they are asymptotically in agreement and converge to the capacity ceiling $C^*$.
\section{Conclusion}
In this work, we investigate the performance analysis of a mixed RF/FSO system with multiple relays employing the amplify-and-forward relaying scheme. Partial relay selection with outdated CSI is adopted to pick one relay for forwarding the signal. Because of its accuracy compared to the Log-Normal and Double-Gamma  distributions, Double-Weibull fading is used as a model of the optical irradiance. We conclude that for moderate turbulence, both the correlation coefficient and the detection method have significant impacts on the system. We observe that the system performs better under the heterodyne mode than IM/DD. We also note that as the time correlation increases, the channel estimation enhances and the performance improves substantially. However, as the correlation becomes very low, the turbulences have no longer impact on the performance and the system depends only on the CSIs correlation. Furthermore, we introduce a general model of hardware impairments to the source and the relays. We conclude that for lower values of the average SNR, the hardware impairments have no observable impacts on the system. However, as the average SNR grows largely, the impairments impact becomes noticeable by a quick saturation of the outage probability and the ergodic capacity. Finally, as an extention of this work, we intend to study the effects of some specified hardware impairments such as the HPA non-linearities and the I/Q imbalance on the mixed RF/FSO system and to quantify the impacts of the different parameters of each hardware impairments on some performance metrics of the system. 
\bibliographystyle{IEEEtran}
\bibliography{main}
\end{document}